\documentclass{article}

\usepackage{amsmath,amssymb,amsfonts,amsthm}
\usepackage{fullpage}
\usepackage{color}

\newtheorem{theorem}{Theorem}[section]
\newtheorem{prop}[theorem]{Proposition}
\newtheorem{lemma}[theorem]{Lemma}
\newtheorem{corollary}[theorem]{Corollary}
\theoremstyle{remark}
\newtheorem{remark}{Remark}

\begin{document}

\title{Algebraic Symmetry and Self--Duality of an Open ASEP}

\author{Jeffrey Kuan}

\date{}

\maketitle

\abstract{We consider the asymmetric simple exclusion process (ASEP) with open boundary condition at the left boundary, where particles exit at rate $\gamma$ and enter at rate $\alpha=\gamma \tau^2$, and where $\tau$ is the asymmetry parameter in the bulk. At the right boundary, particles neither enter nor exit.  By mapping the generator to the Hamiltonian of an XXZ quantum spin chain with reflection matrices, and using previously known results, we show algebraic symmetry and self--duality for the model.}

\section{Introduction}
The asymmetric simple exclusion process (ASEP) is an interacting particle system on a one dimensional lattice, introduced in \cite{Spit70} and \cite{MGP68}. Particles jump one step to the right at rate $p$ and one step to the left at rate $q$, and the jump is blocked if the site is already occupied. As first observed in \cite{HenSch94}, the generator of ASEP can be mapped to the Hamiltonian of the XXZ quantum spin chain, which (with closed boundary conditions) possesses a quantum group symmetry \cite{PasSal90}. Using this symmetry, \cite{Sch97} proves a Markov self--duality for the ASEP with closed boundary conditions. Various modifications and generalizations of this self--duality have since been found \cite{IS},  \cite{BCSDuality}, \cite{BS}, \cite{BS2}, \cite{KuanJPhysA}, \cite{KIMRN},  \cite{BS3}.

A natural extension is to consider open boundary conditions, where particles may enter or exit the lattice. Let $\alpha,\gamma$ denote the entry and exit rates at the left boundary, and let $\beta,\delta$ denote the entry and exit rates at the right boundary. With open boundary conditions, the quantum group symmetry is broken. However, it turns out that for $\alpha/\gamma = p/q, \beta=\delta=0$, a specific algebra element still commutes with the Hamiltonian \cite{Doikou04}. Here, we use this algebra element to show a self--duality for this open ASEP.

\section{Preliminaries} 

\subsection{Spin Chain Notation}
Recall that $\sigma^x,\sigma^y,\sigma^z$ are the Pauli matrices
\begin{equation*}
 \sigma^{x} =\left(\begin{array}{cc}{0} & {1} \\ {1} & {0}\end{array}\right) 
 \quad \quad  
 \sigma^{y} =\left(\begin{array}{cc}{0} & {-i} \\ {i} & {0}\end{array}\right) \quad \quad 
  \sigma^{z} =\left(\begin{array}{cc}{1} & {0} \\ {0} & {-1}\end{array}\right) ,
\end{equation*}
which are a basis for $\mathfrak{sl}_2$, the traceless $2\times 2$ matrices. Also define
\begin{align*}
\sigma^+ &= \frac{\sigma^x + i\sigma^y}{2} =\left(\begin{array}{cc}{0} & {1} \\ {0} & {0}\end{array}\right) ,\\
\sigma^- &= \frac{\sigma^x- i\sigma^y}{2} =\left(\begin{array}{cc}{0} & {0} \\ {1} & {0}\end{array}\right) ,\\
n &= \frac{1-\sigma^z}{2} =\left(\begin{array}{cc}{0} & {0} \\ {0} & {1}\end{array}\right) .
\end{align*}
The subscripts $_j$ indicates that a matrix acts at lattice site $j$. For example, $\sigma^x_j$ acts on $(\mathbb{C}^2)^{\otimes L}$ as $1^{\otimes j-1} \sigma^x \cdots 1^{\otimes L-j}$.
The --ket vector $\left(\begin{array}{c} 0 \\ 1 \end{array}\right)$ corresponds to a \textit{particle} and  $\left(\begin{array}{c} 1 \\ 0 \end{array}\right)$ corresponds to a \textit{hole}. The operators $\sigma_j^-$ and $\sigma_j^+$ are creation and annihilation operators, respectively.

Recall the Yang--Baxter equation
$$
R_{12}(\lambda_1-\lambda_2)R_{13}(\lambda_1 - \lambda_3) R_{23}(\lambda_2 - \lambda_3) = R_{23}(\lambda)R_{12}(\lambda_1)R_{12}(\lambda_1-\lambda_2)
$$
and the reflection equation \cite{Cher84}
$$
R_{12}(\lambda_1-\lambda_2)\mathcal{K}_1(\lambda_1)R_{21}(\lambda_1 + \lambda_2)\mathcal{K}_2(\lambda_2) =  \mathcal{K}_2(\lambda_2) R_{12}(\lambda_1 + \lambda_2)\mathcal{K}_1(\lambda_1)R_{21}(\lambda_1-\lambda_2)
$$
The $R$--matrix of the XXZ model is a one--parameter solution to the Yang--Baxter equation, with the parameter denoted by $\mu$ (in addition to $\lambda$). A solution to the reflection equation is given in \cite{Doikou04}, which has three parameters, denoted $\mu,m,\zeta$ (in addition to $\lambda$). 

The transfer matrix constructed from $R,\mathcal{K}$ leads to a Hamiltonian \cite{Skly88}, which is stated as (1.3) from \cite{Doikou04}:
$$
\begin{aligned} \mathcal{H}=-\frac{1}{4} \sum_{i=1}^{L-1}\left(\sigma_{i}^{x} \sigma_{i+1}^{x}+\sigma_{i}^{y} \sigma_{i+1}^{y}+\cosh i \mu\  \sigma_{i}^{z} \sigma_{i+1}^{z}\right)-\frac{1}{4} \sinh i \mu\left(\sigma_{L}^{z}-\sigma_{1}^{z}\right)-\frac{L+1}{4} \cosh i \mu \\+ \frac{\sinh i \mu}{4 \sinh i \mu\left(\frac{m}{2}+\zeta\right) \cosh i \mu\left(\frac{m}{2}-\zeta\right)} \left(-\sinh(im\mu) \sigma_{1}^{z} + \sigma_{1}^{x}\right)+c_{1}+c_{2} \sigma_{L}^{z}, .\end{aligned}
$$
where $c_1 = (q+q^{-1})^{-1}/2,c_2=0$.
This Hamiltonian is integrable, in the sense that it commutes with a family of transfer matrices: see (2.30) and (4.28) of \cite{Doikou04}.

%When the $K$--matrix is chosen to be
%$$
%\left(\begin{array}{cc}{e^{i \mu \xi}} & {\kappa} \\ {\kappa} & {}\end{array}\right)
%$$
%By (2.17) of \cite{Doikou04}, 
%$$
%\frac{e^{i\mu \xi}}{2\kappa} = -\sinh im \mu, \quad \quad \frac{e^{-i\mu \xi}}{2\kappa} = \sinh 2i\mu\zeta.
%$$ 
%Let $x_1$ depend on $\xi$ and $\kappa$ by 
%$$

%$$

\subsection{Quantum Groups}
Let $A$ be the Cartan matrix of the affine Lie algebra $\widehat{\mathfrak{sl}}_2$
$$
A=\left(\begin{array}{cc}{2} & {-2} \\ {-2} & {2}\end{array}\right)
$$
The Drinfeld--Jimo quantum group $\mathcal{U}_{\tau}(\widehat{\mathfrak{sl}}_2)$ is the bi--algebra generated by $\{e_i,f_i,k_i\}_{i = 1,2}$ with relations
\begin{align*}
k_i k_j = k_j k_i, \quad \quad k_i e_j &= {\tau}^{\tfrac{1}{2} a_{ij}}e_j k_i, \quad \quad k_i f_j = {\tau}^{-\tfrac{1}{2} a_{ij}}f_j k_i \\
[ e_i, f_j] &= \delta_{ij} \frac{k_i^2 - k_i^{-2}}{{\tau}-{\tau}^{-1}}, \quad i,j=1,2,
\end{align*}
and
\begin{align*}
\chi_i^3 \chi_j - ({\tau}^2 + 1 + {\tau}^{-2}) \chi_i^2 \chi_j \chi_i + ({\tau}^2 + 1 + {\tau}^{-2}) \chi_i \chi_j \chi_i^2 - \chi_j \chi_i^3 = 0, \quad \quad \chi_i = {e_i, f_i}, \quad i\neq j. 
\end{align*}
The co--product is defined by
$$
\Delta(\chi_i) = k_i^{} \otimes \chi_i + \chi_i \otimes k_i^{-1}, \quad \chi_i= e_i,f_i, \quad \quad \Delta(k_i^{\pm}) = k_i^{\pm} \otimes k_i^{\pm},
$$
and satisfies the co--associativity
$$
\Delta^{(L)} = (\Delta^{(L-1)} \otimes \mathrm{id}) \circ \Delta= (\mathrm{id} \otimes \Delta^{(L-1)} )\circ \Delta.
$$
The evaluation module $\rho_{\lambda}: \mathcal{U}_{\tau}(\widehat{\mathfrak{sl}}_2) \rightarrow \mathrm{End}(\mathbb{C}^2)$ is given by 
\begin{align*}
\rho_{\lambda}(k_1) &= {\tau}^{\sigma^z/2}, \quad \rho_{\lambda}(e_1) = \sigma^+, \quad \rho_{\lambda}(f_1) = \sigma^-\\
\rho_{\lambda}(k_2) &= {\tau}^{-\sigma^z/2}, \quad \rho_{\lambda}(e_2) = e^{-2\lambda} \sigma^-, \quad \rho_{\lambda}(f_2) = e^{2\lambda} \sigma^+.
\end{align*}
%Note that the inversion $\tau \mapsto \tau^{-1}$ induces an isomorphism from $\mathcal{U}_{\tau}(\widehat{\mathfrak{sl}}_2)$ to $\mathcal{U}_{\tau}(\widehat{\mathfrak{sl}}_2)$ by mapping $e_i$ to $f_i$ and $f_i$ to $e_i$ and $k_i$ to $k_i^{-1}$??? 

\section{Main Results}
Suppose that an ASEP on $L$ lattice sites evolves with right jump rates $p$ and left jump rates $q$ (without assuming $p+q=1$). Particles enter from the left at rate $\alpha$, exit at the right at rate $\beta$, exit at the left at rate $\gamma$, and enter from the right at rate $\delta$. Define the asymmetry parameter
$$
\tau = \sqrt{\frac{p}{q}}.
$$
More precisely, define the generator of ASEP to be the operator\footnote{Here, we use the mathematical physics convention that a stochastic matrix has columns that add up to $1$, rather than its rows.}
\begin{multline*}
\mathcal{L}= - \sqrt{pq}\sum_{j=1}^L  \left( \tau^{-1}( \sigma^-_j \sigma^+_{j+1} - (1-n_j)n_{j+1} ) + \tau(\sigma^-_{j+1}\sigma^+_j - n_j(1-n_{j+1})  ) \right)\\
- \alpha(\sigma_1^- - 1 + n_1) - \gamma(\sigma_1^+-n_1) - \delta(\sigma_L^- - 1  + n_L)-\beta(\sigma_L^+-n_L).
\end{multline*}

In the  ASEP to XXZ change of basis \cite{HenSch94}, the generator of ASEP becomes the XXZ Hamiltonian (see also e.g. equations (2.12)--(2.14) of \cite{SandowPASEP}). More specifically, let $V$ denote the operator
$$
V = \tau^{-\sum_{j=1}^L j n_j}.
$$
Then
\begin{multline*}
V \mathcal{L}V^{-1} = -\frac{1}{2} \sqrt{pq} \sum_{j=1}^{L-1} \left[ \sigma^x_{j}\sigma^x_{j+1} +  \sigma^y_{j}\sigma^y_{j+1} + \frac{\tau + \tau^{-1}}{2} \sigma^z_{j}\sigma^z_{j+1} - \frac{\tau + \tau^{-1}}{2} \right] \\
- A_1^+ \sigma_1^x - i A_1^- \sigma_1^y - B_1 \sigma_1^z - A_L^+ \sigma_L^x - i A_L^- \sigma_L^y - B_L \sigma_L^z + \frac{1}{2}(\alpha + \beta + \gamma + \delta),
\end{multline*}
where
$$
A_1^{\pm} = \frac{1}{2} (\gamma \tau \pm \alpha \tau^{-1}), \quad \quad B_1 = \frac{1}{2}(\gamma - \alpha) + \frac{1}{4}(\tau - \tau^{-1}), \quad \quad A_L^{\pm} = \frac{1}{2}(\beta \tau^L \pm \delta \tau^{-L}), \quad \quad B_L = \frac{1}{2}(\beta - \delta) - \frac{1}{4}(\tau - \tau^{-1}).
$$
Note that
$$
A_1^- = 0 \text{ if and only if } \frac{\alpha}{\gamma} = \frac{ p}{q} \text{ or } \alpha=\gamma=0, \quad \quad A_L^- = 0 \text{ if and only if } \frac{\delta}{\beta} = \left(\frac{p}{q}\right)^L \text{ or } \beta=\delta=0.
$$

\begin{prop}\label{Prop}

Fix $m=-1$, and $\zeta$ such that
$$
\frac{1}{2}(\gamma \tau + \alpha \tau^{-1}) = \frac{\tau - \tau^{-1}}{\tau + (\tau^{-2\zeta} - \tau^{2\zeta}) - \tau^{-1}}.
$$ 
Then there exists a constant $C$ such that
$$
2\mathcal{H} = V \mathcal{L} V^{-1}  + C \mathbb{I}
$$
for the values $\beta=\delta=0$, $\alpha/\gamma = p/q$, $\sqrt{pq} = 1$, $\tau = e^{-i\mu}$.  

\end{prop}
\begin{proof}
We merely need to match the coefficients of the operators in $2\mathcal{H}$ and $V\mathcal{L}V^{-1}$. It is immediate that the coefficients of $\sigma_j^x\sigma_{j+1}^x$ and $\sigma_j^y\sigma_{j+1}^y$ match.
Setting $\tau = e^{-i\mu}$, we have
$$
\cosh i\mu = \frac{e^{i\mu} + e^{-i\mu}}{2} = \frac{\tau + \tau^{-1}}{2}, \quad \quad \sinh i\mu =  \frac{e^{i\mu} - e^{-i\mu}}{2} = -\frac{\tau - \tau^{-1}}{2},
$$
showing that the coefficient of $\sigma_j^z \sigma_{j+1}^z$ matches. If furthermore, $\beta=\delta=0$, then $A_L^{\pm}=0$ and there is no $\sigma_L^x,\sigma_L^y$ contribution, so the $\sigma_L^x,\sigma_L^y$ coefficient matches. For $\alpha = (p/q)\gamma = \gamma \tau^2$, the term $A_1^-$ equals zero, so the $\sigma^y_1$ coefficient matches. Note that when $c_2=0$, the coefficient of $\sigma_L^z$ in $2\mathcal{H}$ becomes
$$
-\frac{1}{2} \sinh i\mu =  \frac{1}{4}(\tau - \tau^{-1}),
$$
which equals $-B_L$ with $\beta=\delta$. Thus the coefficient of $\sigma_L^z$  matches. 

It remains to match the coefficients of $\sigma_1^x$ and $\sigma_1^z$, meaning that there are two equalities to show So comparing the $\sigma_1^z$ terms, it remains to show
$$
-B_1 = -\frac{1}{2}(\gamma -\alpha) - \frac{1}{4}(\tau - \tau^{-1}) = \frac{1}{2}\sinh i\mu - \frac{\sinh i \mu\  \sinh (i m \mu)}{2 \sinh i \mu\left(\frac{m}{2}+\zeta\right) \cosh i \mu\left(\frac{m}{2}-\zeta\right)}
$$
and comparing the $\sigma_1^x$ terms, it remains to show
$$
-A_1^+ = -\frac{1}{2}(\gamma \tau + \alpha \tau^{-1}) = \frac{\sinh i \mu}{2 \sinh i \mu\left(\frac{m}{2}+\zeta\right) \cosh i \mu\left(\frac{m}{2}-\zeta\right)}.
$$
 In the former equality, the term $\frac{1}{2}\sinh i\mu$ cancels $-\frac{1}{4}(\tau-\tau^{-1})$, so we are left to show
$$
\frac{1}{2}(\gamma - \alpha) = \frac{\sinh i \mu\  \sinh (i m \mu)}{2 \sinh i \mu\left(\frac{m}{2}+\zeta\right) \cosh i \mu\left(\frac{m}{2}-\zeta\right)} .
$$

Now, using that $m=-1$, we have
$$
\sinh(i m\mu) = - \frac{\tau^m - \tau^{-m}}{2} = - \frac{1-\tau^2}{2\tau}= - \frac{\gamma - \alpha}{2\gamma \tau}  = - \frac{\gamma - \alpha}{\gamma \tau + \alpha \tau^{-1}}, 
$$
which means that the two equalities are equivalent to each other. So it just remains to show that

\begin{align*}
\frac{1}{2}(\gamma \tau + \alpha \tau^{-1}) & = - \frac{\tau-\tau^{-1}}{ (\tau^{-1/2}e^{i\mu\zeta} - \tau^{1/2}e^{-i\mu\zeta})(\tau^{-1/2}e^{-i\mu\zeta} + \tau^{1/2}e^{i\mu\zeta} )   }\\
&= - \frac{\tau - \tau^{-1}}{ \tau^{-1} + (e^{2i\mu\zeta} - e^{-2i\mu\zeta}) - \tau }\\
&= \frac{\tau - \tau^{-1}}{\tau + (\tau^{-2\zeta} - \tau^{2\zeta}) - \tau^{-1}},
\end{align*}
which we have assumed to be true.

\end{proof}

\begin{remark}
In \cite{Ligg75}, it shown that on the half--line, stationary measures exist when $\alpha/p + \gamma/q=1$. A phase transition occurs at $\alpha/p=1/2$: for $\alpha/p<1/2$, there exist stationary measures with i.i.d. Bernoulli random variables with parameter $\alpha/p$, and when $\alpha/p>1/2$ the stationary measures are spatially correlated. Under the additional condition that $\alpha/p + \gamma/q=1$, the condition $\alpha/\gamma=p/q$ is equivalent to $\alpha/p = \gamma/q=1/2$.
\end{remark}

\begin{remark}
The condition $\alpha/\gamma = p/q$ had previously appeared in \cite{BBCWDuke}, which considered $\alpha/p=\gamma/q=1/2$. Note that a variant of the reflection equation is satisfied in the stochastic vertex model of \cite{BBCWDuke} -- see Propositions 4.3 and 4.10 in that reference.
\end{remark}

\begin{remark}
The open boundary conditions here are different than the one considered in \cite{KuanSF}, which proves duality (but not self--duality) for the ASEP without algebraic considerations. 
\end{remark}

\begin{lemma}\label{Lem}
For the values $\beta=\delta=0$, we have the detailed balance condition
$$
V^{2} \mathcal{L} V^{-2} = \mathcal{L}^*, 
$$
where the $^*$ denotes the transposition.
\end{lemma}
\begin{proof}
\underline{Proof 1:}

The Hamiltonian $\mathcal{H}$ is Hermitian\footnote{Although this is not explicitly stated in \cite{Doikou04}, Hamiltonians in mathematical physics are always Hermitian.}, meaning that
$$
\mathcal{H}^* = \mathcal{H}.
$$
Therefore, by Proposition \ref{Prop},
\begin{align*}
\mathcal{L}^* + C \mathbb{I} &= 2V \mathcal{H}^* \mathcal{V}^{-1} \\
&= 2V \mathcal{H} V^{-1} \\
&=V^2 \mathcal{L} V^{-2} + C\mathbb{I},
\end{align*}
implying the lemma.

\underline{Proof 2:}

Note that \cite{CorteelWilliamsDuke} gives the stationary measures for open ASEP with generic $\alpha,\beta,\gamma,\delta$. For these generic parameters, the process is not reversible. However, for our choice of parameters ($\tau^2=\alpha/\gamma$) the process does turn out to be reversible. Let $\vert \eta^+\rangle$ and $\vert \eta^- \rangle$ be two basis vectors which only different at the left boundary, where $\eta^+$ has a particle and $\eta^-$ does not. Then
\begin{align*}
\langle \eta^- \vert \mathcal{L} \vert \eta^+ \rangle &= \gamma = \frac{\alpha}{\alpha/\gamma} = \langle \eta^- \vert \mathcal{L}^* \vert \eta^+\rangle \frac{\langle \eta^- \vert V^{-2} \vert \eta^-\rangle}{\langle \eta^+ \vert V^{-2} \vert \eta^+\rangle,}\\
\langle \eta^+ \vert \mathcal{L} \vert \eta^- \rangle &= \alpha = \frac{\gamma}{(\alpha/\gamma)^{-1}}= \langle \eta^+ \vert \mathcal{L}^* \vert \eta^-\rangle \frac{\langle \eta^+ \vert V^{-2} \vert \eta^+\rangle}{\langle \eta^- \vert V^{-2} \vert \eta^-\rangle},\
\end{align*}
which shows the detailed balance equation at the boundary. In the bulk, the detailed balance equation reduces to teh detailed balance equation for ASEP with closed boundaries, which is known to hold.
\end{proof}

%\begin{remark}
%If $\alpha/\gamma = p/q$, as in Proposition \ref{Prop}, then $W$ becomes
%$$
%\tau^{2\sum_{j=1}^L n_j}\tau^{2\sum_{j=1}^L jn_j}.
%$$
%For closed boundary conditions, the diagonal operator $\sum_{j=1}^L n_j$ becomes a constant, equal to the number of particles.
%\end{remark}

Let $\mathcal{Q}^1(s) \in \mathcal{U}_{\tau}(\widehat{\mathfrak{sl}}_2)$ be the element from (4.1) of \cite{Doikou04}
$$
\mathcal{Q}^1(s) = s^{-1} k_1e_1 + sk_1f_1 + x_1k_1^2 - x_1 \mathbb{I}.
$$ 
This element satisfies the property that
For $m=-1$ the value of $x_1 = \frac{e^{i \mu \xi}}{2\kappa \sinh i\mu}$ is simply equal to $1$, by  (2.17) of \cite{Doikou04}.  By (4.30) of  \cite{Doikou04}\footnote{The paper \cite{Doikou04} uses a different co--product than the one here: the left and right tensor products are reversed. This is due to the choice of the direction of asymmetry in the ASEP; here, we have an open boundary at the left and a closed boundary at the right, whereas the choice of reflection matrices in \cite{Doikou04} would have a closed boundary at the left (diagonal reflection matrix) and an open boundary at the right (non--diagonal reflection matrix). The examples in section \ref{Ex} will demonstrate that this is the correct choice of co--product for our present case.}, there is the commutation
$$
[\mathcal{H}, \rho_0^{\otimes L}(\Delta^{(L)}(\mathcal{Q}^1(\tau^{-1/2})))]=0.
$$

For any value of $\lambda$, the evaluation representation $\rho_{\lambda}$ maps $\mathcal{Q}^1(\tau^{-1/2})$ to 
\begin{equation}\label{M}
\left( 
\begin{array}{cc}
\tau^{-1}  -1 & 1 \\
1 & \tau-1
\end{array}
\right).
\end{equation}
By the relations in the quantum group $ \mathcal{U}_{\tau}(\widehat{\mathfrak{sl}}_2)$,
$$
 \mathcal{Q}^1(s\tau^{-1})k_1^2 = k_1^2 \mathcal{Q}^1(s).
$$
One can see directly that
\begin{equation}\label{abc}
\Delta( \mathcal{Q}^1(\tau^{-1/2}) ) = k_1^2 \otimes \mathcal{Q}^1(\tau^{-1/2})  + \mathcal{Q}^1(\tau^{-1/2})  \otimes 1,
\end{equation}
so by co--associativity
$$
\Delta^{(L)}( \mathcal{Q}^1(\tau^{-1/2}) ) = \sum_{x=1}^L \underbrace{k_1^2 \otimes \cdots \otimes k_1^2}_{x-1} \otimes \mathcal{Q}^1(\tau^{-1/2})  \otimes \underbrace{ 1 \otimes \cdots  \otimes 1}_{L-x}
$$

Let $S_N$ be the operator
$$
S_N:= \rho_0^{\otimes L}(\Delta^{(L)}(\mathcal{Q}^1(\tau^{-1/2}))^N)
$$
and let $D_N$ be the operator
$$
D_N = V^{}S_NV^{}.
$$

\begin{theorem}
For any $N\geq 1$, we have the duality result
$$
\mathcal{L}^*D_N = D_N \mathcal{L}.
$$
\end{theorem}
\begin{proof}
Once Proposition \ref{Prop} and Lemma \ref{Lem} are proven, this is similar to the argument made in \cite{CGRS}. We briefly recall the proof again for completeness.

Combine the two identities
$$
2\mathcal{H} = V\mathcal{L}V^{-1} + C \mathbb{I}
$$
and
$$
\mathcal{H}S_N = S_N\mathcal{H},
$$
to get that
$$
V \mathcal{L} V^{-1} S_N + CS_N = S_N V \mathcal{L} V^{-1} + CS_N.
$$
Now, using that
$$
V^{2} \mathcal{L}V^{-2} = \mathcal{L}^*,
$$
we have
$$
V^2 \mathcal{L} V^{-2} VS_N  = V S_N V^{} \mathcal{L} V^{-1}
$$
is equivalent to
$$
\mathcal{L}^*  V S_N V = V S_N V \mathcal{L}.
$$
\end{proof}

\begin{remark}
Note that the duality function $D_N$ and symmetry function $S_N$ only depends on $\alpha$ and $\gamma$ through their ratio $\alpha/\gamma$.
\end{remark}

By applying additional symmetries, we obtain two more duality functions. Let $\Pi$ be the particle hole involution, defined by 
$$
\Pi = 
\left(
\begin{array}{cc}
0 & 1 \\
1 & 0
\end{array}
\right)^{\otimes L}.
$$
Let $\widetilde{\mathcal{L}}$ be the generator for ASEP, where particles jump to the \textit{left} at rate $p$ and \text{right} at rate $q$, particles \textit{exit} at the left boundary at rate $\alpha$ and enter at the left boundary at rate $\gamma$, with closed boundary conditions at the right boundary. 

\begin{corollary} We have
$$
\mathcal{L}^*D_N\Pi = D_N\Pi \widetilde{\mathcal{L}}.
$$
On the semi--infinite lattice $\mathbb{Z}_{>0}$, we have
$$
\mathcal{L}^*DV^{-2} = DV^{-2} \widetilde{\mathcal{{L}}}.
$$

\end{corollary}
\begin{proof}
The holes of ASEP evolving under $\mathcal{L}$ have the same evolution as the particles of ASEP evolving under $\tilde{\mathcal{L}}$. In other words, $\Pi \mathcal{L} \Pi = \mathcal{L}^*$. This implies the first statement. It can be checked directly that on the semi--infinite lattice, $V^{-2} \widetilde{\mathcal{L}} V^2 = \mathcal{L}$, implying the second statement.
\end{proof}

We now proceed to an explicit expression for $S_N$ (and hence of $D_N$). For non--negative integers $a$ and $b$, let $Q^{a,b}$ be the element of $\mathcal{U}_{\tau}(\widehat{\mathfrak{sl}}_2)$ defined by 
$$
Q^{a,b} = \sum_{l_1 + \ldots + l_a \leq b}  (k_1^2)^{b} \mathcal{Q}^1(\tau^{l_1+\ldots + l_{a}-1/2})  \cdots \mathcal{Q}^1(\tau^{l_2+\ldots + l_{a}-1/2})  \cdots \mathcal{Q}^1(\tau^{l_a-1/2}) 
$$
For any set of integers $m_1,\ldots,m_L$ and $1 \leq i\leq j \leq L$, let $m_{[i,j]}=m_i + \ldots + m_j$.

\begin{prop}
The symmetry operator $S_N$ has the form
$$
\sum_{m_1+\ldots+m_L=N} \sum_{j=1}^L \rho_0(Q^{m_j,m_{[j+1,L]}}_j )
$$
\end{prop}
\begin{proof}
When expanding $\Delta^{(L)}(\mathcal{Q}^1(\tau^{-1/2}))^N$, let $m_j$ denote the number of times that $\mathcal{Q}^1$ acts on lattice site $j$ for $1 \leq j \leq L$. We must have that $m_1 + \ldots + m_L=N$. At lattice site $x$, the operator $k_1^2$ acts $m_{[j+1,L]}$ times, corresponding to the $m_{[j+1,L]}$ times that $\mathcal{Q}^1$ acts to the right of $j$. So the action at lattice site $j$ is of the form
$$
k_1^2 \cdots k_1^2 \mathcal{Q}^1(\tau^{-1/2}) k_1^2 \cdots  \cdots  k_1^2\mathcal{Q}^1(\tau^{-1/2})  k_1^2 \cdots k_1^2 \mathcal{Q}^1(\tau^{-1/2})  k_1^2 \cdots k_1^2.
$$
Let $l_0$ denote the length of the first block of $k_1^2$, and let $l_1$ denote the length of the second block, and so forth, up to $l_{m_j}$. We must have
$
l_0 + \ldots + l_{m_j} = m_{[j+1,L]}.
$
By repeated applications of \eqref{abc}, the result follows. 
\end{proof}

\section{Examples}\label{Ex}
Suppose that $L=1$ and $N$ is arbitrary. Let $\vert 1\rangle$ denote the vector $(0\ 1)$ and $\vert 0\rangle$ denote the vector $(1\ 0)$. Taking $M$ to be the matrix in \eqref{M}, the identity $\langle 0 \vert\mathcal{L}^*D \vert 1 \rangle= \langle 0 \vert D \mathcal{L} \vert 1\rangle$ becomes
$$
-\gamma (M^N)_{12} \tau^{-1} + \gamma (M^N)_{11} = -\alpha (M^N)_{21} \tau^{-1} + \alpha (M^N)_{22}\tau^{-2} .
$$
For $\alpha = \gamma \tau^2$, one can check that both sides equal $-\gamma$ for $N$ odd and $\gamma$ for $N$ even. 

Suppose that $N=1$ and $L$ is arbitrary. Let $\vert x \rangle$ denote the particle configuration with a single particle at site $x$, and $\vert \emptyset \rangle$ denote the particle configuration with no particles. When $\mathcal{Q}^1$ is applied to lattice site $y$, the operator $k_1^2$ acts on the $y-1$ sites to the left, as the constant $\tau^{}$ on particles and $\tau^{-1}$ on holes.  The operator $V = \tau^{-\sum_j j n_j}$ acts on all sites, but only has a nonzero contribution at particles. Thus $0=\langle \emptyset \vert\mathcal{L}^*D \vert x \rangle= \langle \emptyset \vert D \mathcal{L} \vert x\rangle$ amounts to the identity
$$
p (\tau^{-1})^x (\tau^{-1})^{L-(x+1)}  + q (\tau^{-1})^{x-2}(\tau^{-1})^{L-(x-1)}  - (p+q)(\tau^{-1})^{x-1}(\tau^{-1})^{L-x} =0.
$$
And indeed, for $\tau = \sqrt{p/q}$, the left--hand--side is
$$
\tau^{-(L-1)}(p+q- p -q),
$$
which equals $0$.

\bibliographystyle{alpha}
\bibliography{Exposition}
\end{document}